\documentclass{llncs}
\usepackage{tikz,tkz-graph,amsmath,amsfonts,amssymb,graphics,color,comment,graphicx,boxedminipage}
\usetikzlibrary{arrows.meta}

\newcommand{\NP}{{\sf NP}}
\newcommand{\FPT}{{\sf FPT}}

\newtheorem{observation}{Observation}

\newcommand{\problemdef}[3]{
	\begin{center}
		\begin{boxedminipage}{.99\textwidth}
			\textsc{{#1}}\\[2pt]
			\begin{tabular}{ r p{0.8\textwidth}}
				\textit{~~~~Instance:} & {#2}\\
				\textit{Question:} & {#3}
			\end{tabular}
		\end{boxedminipage}
	\end{center}
}

\begin{document}

\title{On the Parameterized Complexity of\\ $k$-Edge Colouring  \thanks{The second author was supported by  the Research Council of Norway via the project CLASSIS and the work was done when the second author was visiting the University of Fribourg funded by a scholarship of the University of Fribourg.
The third author was supported by the Leverhulme Trust (RPG-2016-258).}}

\author{
Esther Galby\inst{1}
\and 
Paloma T. Lima\inst{2}
\and
Dani{\"e}l Paulusma\inst{3}
\and
Bernard Ries\inst{1}
}

\institute{
Department of Informatics, University of Fribourg, Switzerland,  \texttt{\{esther.galby,bernard.ries\}@unifr.ch}
\and
Department of Informatics, University of Bergen, Norway, \texttt{paloma.lima@uib.no}
\and
Department of Computer Science, Durham University, UK, \texttt{daniel.paulusma@durham.ac.uk}
}

\maketitle

\begin{abstract}
For every fixed integer~$k\geq 1$, we prove that  $k$-{\sc Edge Colouring} is fixed-parameter-tractable when parameterized by the number of vertices of maximum degree.
\end{abstract} 

\section{Introduction}\label{s-intro}

For an integer $k\geq 1$, a {\it $k$-edge colouring} of a graph $G=(V,E)$ is a mapping $c:E\to \{1,\ldots,k\}$ such that
$c(e)\neq c(f)$ for any two edges $e$ and $f$ of $G$ that have a common end-vertex.
The {\sc Edge Colouring} problem is to decide if a given graph $G$ has a $k$-edge colouring for some given integer~$k$.
If $k\geq 1$ is {\it fixed}, that is, not part of the input, then we denote the problem as:

\problemdef{$k$-Edge Colouring}{a graph $G$.}{does $G$ have a $k$-edge colouring?} 
The {\it chromatic index} of a graph $G=(V,E)$ is the smallest integer~$k$ such that $G$ has a $k$-edge colouring. 
The {\it degree} $d_G(u)$ of a vertex~$u\in V$ is the size of its neighbourhood $N(u)=\{v\; |\; uv\in E\}$. We let $\Delta=\Delta_G$ denote the maximum degree of $G$. Vizing showed the following classical result.

\begin{theorem}[\cite{Vi64}]\label{t-vizing}
The chromatic index of a graph $G$ is either $\Delta$ or $\Delta+1$.
\end{theorem}
Due to Theorem~\ref{t-vizing}, we can make the following observation.

\begin{observation}\label{o-max}
Let $G$ be a graph. If $\Delta\leq k-1$, then  $G$ is a yes-instance of {\sc $k$-Edge Colouring}. If $\Delta\geq k+1$, then $G$ is a no-instance of {\sc $k$-Edge Colouring}.
 \end{observation}
 
 By Observation~\ref{o-max} we may assume without loss of generality that an input graph of {\sc $k$-Edge Colouring} has maximum degree~$k$.
The following well-known hardness result was proven by Holyer for $k=3$ and by Leven and Galil for $k\geq 4$.

\begin{theorem}[\cite{Ho81,LG83}]\label{t-hard}
For every $k\geq 3$, {\sc $k$-Edge Colouring} is \NP-complete even for graphs in which \emph{every} vertex has degree~$k$. 
\end{theorem}

In this note we consider the $k$-{\sc Edge Colouring} problem from the viewpoint of Parameterized Complexity, where problem inputs are specified by a main part $I$ of size~$n$ and a parameter $p$, which is assumed to be small compared to $n$. The main question is to determine if a problem is \emph{fixed-parameter tractabe} (\FPT), that is, if it can be solved in time 
$f(p)n^{O(1)}$, where $f$ is a computable (and possibly exponential) function of $p$. The choice of parameter depends on the context. 
By the above discussion, the most natural choice of parameter for $k$-{\sc Edge Colouring} is the number of vertices of maximum degree. We prove the following result (note that $k$-{\sc Edge Colouring} is polynomial-time solvable for
$k\leq 2$). 

\begin{theorem}\label{t-main}
For every $k\geq 1$, $k$-{\sc Edge Colouring} can be solved in 
 $O(p^2k^{\frac{k^2p}{2}+4}+n+m)$ time 
 on graphs with $n$ vertices, $p$ of which have maximum degree, and $m$ edges.
 Moreover, it is possible to find a $k$-edge colouring of $G$ in $O(p^2k^{\frac{k^2p}{2}+4}+k^2(n-p)+n+m)$ time (if it exists).
\end{theorem}

As we assume that $k$ is a fixed constant, Theorem~\ref{t-main} implies that for every $k\geq 3$, the $k$-{\sc Edge Colouring} problem is fixed-parameter tractable when parameterized by the number of vertices of maximum degree.
We prove Theorem \ref{t-main} in Section~\ref{s-main} by using the alternative proof of Theorem \ref{t-vizing} given by
Ehrenfeucht, Faber and Kierstead~\cite{EFK84}. We first discuss some related work.

\subsection*{Related Work}

Apart from a number of (classical) complexity results (e.g.~\cite{Bo90,CE91,AMM15,SG07,GLPR,MFT13,Ma14,Ma17,OMS98}), most results on edge colouring are related to Theorem~\ref{t-vizing} and are of a more structural nature. That is, they involve the derivation of sufficient or necessary conditions for a graph to be $\Delta$-edge colourable  (see, for example~\cite{CFK95,CH89,HZ93,LGZA15}), in which case the graph is said to be Class~1 (and Class~2 otherwise). This is also the focus of papers on edge colouring related to the number of maximum-degree vertices (see, for example,~\cite{CH84b,Ch96,NV90,Ya86}). In particular, 
Fournier~\cite{Fo77} proved that every graph~$G$ in which the vertices of degree~$\Delta$ induce a forest is Class~1.
As a consequence, a graph with at most two vertices of maximum degree is Class~1.
In~\cite{CH84} and~\cite{CH85}, Chetwynd and Hilton gave necessary and sufficient polynomial-time verifiable conditions for a graph with three, respectively, four vertices of maximum degree to be Class~1. The same authors proved an analogous result for $n$-vertex graphs with $r$ vertices of maximum degree in~\cite{CH90}, assuming $\Delta\geq \lfloor\frac{n}{2}\rfloor+\frac{7}{2}r-3$.

For more on edge colouring we refer to the recent survey of Cao, Chen, Jing, Stiebitz and Toft~\cite{CCJST}.

\section{The Proof of Theorem~\ref{t-main}}\label{s-main}

In this section we prove Theorem~\ref{t-main}. 
Let $G=(V,E)$ be a graph and let $S\subseteq V$.
The graph $G[S]=(S,\{uv\in E(G)\; |\; u,v\in S\})$ is the subgraph of $G$ {\it induced} by~$S$. 
The set $N(S) = \bigcup_{u \in S} N(u)$ is the set of vertices adjacent to at least one vertex in $S$.
Let $X$ be the set of maximum-degree vertices of $G$. In this context, the graph~$G[X]$ is said to be the {\it core} of $G$, whereas the graph $G[X\cup N(X)]$ is said to be the {\it semi-core} of $G$.

Machado and de~Figueiredo~\cite{MF10} proved the following result.

\begin{theorem}[\cite{MF10}]\label{t-mf10}
Let $G$ be a graph and $k\geq 0$ be an integer.
Then $G$ has a $k$-edge colouring if and only if its semi-core has a $k$-edge colouring. 
\end{theorem}

The proof of Theorem~\ref{t-mf10} is based on an application of the recolouring procedure of Vizing~\cite{Vi64} for proving Theorem~\ref{t-vizing} and which was also used by Misra and Gries~\cite{MG92} for giving a constructive proof of Theorem~\ref{t-vizing}. 
As explained by Zatesko~\cite{Za18}, the proof of Theorem~\ref{t-mf10} immediately yields a polynomial-time algorithm
for finding a $k$-edge colouring of a graph $G$ given a $k$-edge colouring of its semi-core (if it exists).
Below we give a new, short algorithmic proof of Theorem~\ref{t-mf10}, with the same time complexity,
via a modification of the alternative proof of Theorem~\ref{t-vizing} by Ehrenfeucht et al.~\cite{EFK84}, which might be of independent interest. We state this result as a lemma, as we use it to prove our main result.

\begin{lemma}\label{l-structural}
Let $G$ be a graph with a core of size~$p$.
Given a $k$-edge colouring of its semi-core, it is possible to construct in time $O(k^2(n-p))$ a $k$-edge colouring of $G$. 
\end{lemma}

\begin{proof}
Recall that $X$ denotes the set of vertices of maximum degree~$\Delta$ of $G$.
Let $c'$ be a $k$-edge colouring of the semi-core $G[X\cup N(X)]$ of $G$ and denote by $q$ the size of the semi-core of $G$. Note that $\Delta\leq k$.
If $V \backslash (X \cup N(X))=\emptyset$, then $c'$ is a $k$-edge colouring of $G$. Assume that $V \backslash (X \cup N(X))\neq \emptyset$.

We write $V \backslash (X \cup N(X))=\{u_1,\ldots,u_{n-q}\}$. We let $V_0=X\cup N(X)$ and $V_i=X\cup N(X)\cup \{u_1,\ldots,u_i\}$ for $i=1,\ldots,n-q$. Note that $G=G[V_{n-q}]$.

We define $c_0=c'$. We now show how to extend $c_0$ vertex by vertex until we obtain a $k$-edge colouring of the whole graph.
That is, for $i=1,\ldots,n-q$, we show how to construct a $k$-edge colouring~$c_i$ of $G[V_i]$ given a $k$-edge colouring~$c_{i-1}$ of $G[V_{i-1}]$ 

Thus suppose that $i\geq 1$ and suppose that we already have a $k$-edge colouring $c_{i-1}$ of $G[V_{i-1}]$ (note that at the start of the procedure, when $i=1$, this is indeed the case, as we have $c_0$).
We say that a vertex $u \in V_{i-1} \cup \{u_i\}$ \textit{misses} colour $\ell$ 
if none of the edges incident to $u$ is coloured $\ell$ by $c_{i-1}$. We denote the set of colours that $u$ misses by $F(u)$. Note that $F(u_i) = \{1,\ldots, k\}$. We write $F(uv)= F(u) \cap F(v)$. We now iteratively colour the edges incident to $u_i$ in $G[V_i]$ in such a way that at any time, the resulting mapping is a partial $k$-edge colouring of $G[V_i]$ that satisfies the two following properties:
\begin{itemize}
\item[(1)] for each uncoloured edge $u_iv$ in $G[V_i]$, we have $F(u_iv) \neq \emptyset$;
\item[(2)] there is at most one uncoloured edge $u_iv$ in $G[V_i]$ with $|F(u_iv)| = 1$.
\end{itemize}
Observe that at the start of this process, when no edge $u_iv$ in $G[V_i]$ has been coloured, $c_{i-1}$ satisfies~(1) and~(2). This can be seen as follows. Vertex $u_i$ is in $G[V_i]$ only adjacent to vertices of $N(X)\cup \{u_1,\ldots,u_{i-1}\}$ (where we define $\{u_1,\ldots,u_{i-1}\}=\emptyset$ if $i=1$).
Every $v\in N(X)\cup \{u_1,\ldots,u_{i-1}\}$ that is adjacent to $u_i$ does not belong to $X$ and thus has degree at most $k-1$. Hence, $v$ is adjacent to at most $k-2$ vertices in $V_i\setminus \{u_i\}$ and thus has $|F(v)|\geq 2$. As $|F(u_i)|=\{1,\ldots,k\}$, this means that each edge
$u_iv$ in $G[V_i]$ has $|F(u_iv)|\geq 2$, implying that (1) and (2) are satisfied.

We will now describe how we can maintain properties (1) and (2) while colouring the edges incident to $u_i$ in $G[V_i]$ one by one.
Throughout this process we maintain a set $W =\{v\in V_{i-1}~|~ u_iv \text{ is an uncoloured edge in $G[V_i]$}\}$, so at the start $W$ consists of all neighbours of $u_i$ in $G[V_i]$. We distinguish two cases.

\medskip
\noindent 
{\bf Case 1.} There exists a colour $\ell \in \bigcup_{v \in W} F(u_iv)$ such that $\ell$ appears in at most one set $F(u_iv)$ for some $v \in W$ with $|F(u_iv)| \leq 2$. Then we assign colour $\ell$ to the edge $u_iv$ for which $|F(u_iv)|$ is the smallest over all sets $F(u_iv)$ that contain~$\ell$. If all these sets have size at least~3, then afterwards (1) and (2) are satisfied. Otherwise, there is a unique smallest set $F(u_iv)$ of size at most~2, which also implies that afterwards (1) and (2) are satisfied. 

\medskip
\noindent
{\bf Case 2.} For each colour~$x \in \bigcup_{v \in W} F(u_iv)$, there exists at least two distinct vertices $w,w' \in W$ 
such that $x \in F(u_iw) \cap F(u_iw')$ and both $F(u_iw)$ and $F(u_iw')$ have size at most~2. This means that
\[  2 \left|\bigcup\limits_{v\in W} F(u_iv)\right| \leq \sum\limits_{v \in W: |F(u_iv)| \leq 2} |F(u_iv)| \leq 2|W|. \]

\medskip
\noindent
Hence, we deduce that $|\bigcup_{v\in W} F(u_iv)| \leq |W|$.
As $u_i$ has degree at most~$k-1$ in $G$, it follows that in $G[V_i]$, at most $k-1-|W|$ edges incident to $u_i$ have already been coloured in this stage. It follows that $|F(u_i)| \geq k-(k-1-|W|)=|W| + 1$. Hence, there exists a colour $b \in F(u_i) \backslash \bigcup_{v \in W} F(u_iv)$. Consequently, each vertex in $W$ must have an incident edge coloured $b$. Now choose a vertex $w \in W$ for which $|F(u_iw)|$ is minimum and consider a colour $a \in F(u_iw)$. We swap colours $a$ and $b$ along the path~$P$ in $G[V_i]$ that starts in $w$ and whose edges are alternatingly coloured $a$ and $b$. This yields another partial $k$-edge colouring of $G[V_i]$. However, now $w$ has no incident edge coloured~$b$ anymore, which means that we can colour the edge $u_iw$ with colour $b$. We observe that for all $v\in W\setminus \{w\}$, the set $F(u_iv)$ remains unchanged except the end-vertex $w^*$ of $P$ if $w^*$ is adjacent to $u_i$; in that case $F(u_iw^*)$ is replaced by $F(u_iw^*) \backslash \{a\}$. However, by minimality of $|F(u_iw)|$, we have that $|F(u_iw)| \leq |F(u_iw^*)|$ and property~(2) implies that if $|F(u_iw)| = 1$, then $|F(u_iw^*)| \geq 2$. Thus, $|F(u_iw^*) \backslash \{a\}| \geq 1$. We conclude that both~(1) and~(2) are satisfied by the newly obtained (partial) $k$-edge colouring of $G[V_i]$.

\medskip
\noindent
After colouring the last uncoloured edge in $G[V_i]$ incident to~$u_i$ we have obtained our $k$-edge colouring~$c_i$ of $G[V_i]$.
Hence, after doing this for $i=n-q$, we have indeed extended $c'$ to a $k$-edge colouring~$c=c_{n-q}$ of $G=G[V_{n-q}]$. 

Finally note that since $q \geq p$, the algorithm has at most $n-p$ extension steps and as each takes $O(k^2)$ time, the total running time is $O(k^2(n-p))$.
This completes the proof of Lemma~\ref{l-structural}.
\qed
\end{proof}

We are now ready to prove Theorem~\ref{t-main}, which we restate below.

\medskip
\noindent
{\bf Theorem~\ref{t-main}.}
{\it For every $k\geq 1$, $k$-{\sc Edge Colouring} can be solved in 
 $O(p^2k^{\frac{k^2p}{2}+4}+n+m)$ time on graphs with $n$ vertices, $p$ of which have maximum degree, and $m$ edges. Moreover, it is possible to find a $k$-edge colouring of $G$ in $O(p^2k^{\frac{k^2p}{2}+4}+k^2(n-p)+n+m)$ time (if it exists).}

\begin{proof}
Let $G=(V,E)$ be an instance of $k$-{\sc Edge Colouring}.
We first compute the maximum degree~$\Delta$ of $G$ and the corresponding set~$X$ of vertices of degree~$\Delta$ in $O(n+m)$ time (so $p=|X|)$.
Let $H=G[X\cup N(X)]$ denote the semi-core of $G$.
By Theorem~\ref{t-mf10}, it suffices to check if $H$ has a $k$-edge colouring. Note that we may assume $\Delta\leq k$ by
Observation~\ref{o-max}, hence
$|N(X)|\leq \Delta p\leq kp$. This means that
\[2|E(H)|=  \sum\limits_{u \in X} d_H(u) + \sum\limits_{u \in N(X)} d_H(u) \leq \sum\limits_{u \in X} d_G(u) + \sum\limits_{u \in N(X)} d_G(u) \leq kp + kp(k-1).\] 
Hence, $H$ has at most $\frac{k^2p}{2}$ edges. Checking if a mapping $f:E(H)\to \{1,\ldots,k\}$ is a $k$-edge colouring takes
$O(k^4p^2)$ time. The number of such mappings is at most $k^{\frac{k^2p}{2}}$.
Hence, brute force checking if $H$ has a $k$-edge colouring takes time $O(p^2k^{\frac{k^2p}{2}+4})$.
Thus the first statement of the theorem follows.
We obtain the second statement by applying Lemma~\ref{l-structural}. \qed
\end{proof}

\section{Conclusions}\label{s-con}

We proved that $k$-{\sc Edge Colouring} is fixed-parameter-tractable when parameterized by the number of maximum-degree vertices.
We note that our proof does not work to show this for {\sc Edge Colouring}, as the set $X\cup N(X)$ may have size $\Omega(k)$.
As such, we pose the following open problem: what is the computational complexity of {\sc Edge Colouring} when parameterized by the number of maximum-degree vertices?

\bigskip
\noindent
{\it  Acknowledgements.}
We thank Leandro Zatesko for bringing paper~\cite{MF10} to our attention.


\begin{thebibliography}{10}

\bibitem{Bo90}
H.~L. Bodlaender.
\newblock Polynomial algorithms for graph isomorphism and chromatic index on
  partial $k$-trees.
\newblock {\em Journal of Algorithms}, 11(4):631--643, 1990.

\bibitem{CE91}
L.~Cai and J.~A. Ellis.
\newblock {N}{P}-completeness of edge-colouring some restricted graphs.
\newblock {\em Discrete Applied Mathematics}, 30(1):15--27, 1991.

\bibitem{CCJST}
Y.~Cao, G.~Chen, G.~Jing, M.~Stiebitz, and B.~Toft.
\newblock Graph edge coloring: {A} survey.
\newblock {\em Graphs and Combinatorics}, to appear.

\bibitem{CFK95}
B.-L. Chen, H.-L. Fu, and M.~T. Ko.
\newblock Total chromatic number and chromatic index of split graphs.
\newblock {\em Journal of Combinatorial Mathematics and Combinatorial
  Computing}, 17:137--146, 1995.

\bibitem{CH84b}
A.~G. Chetwynd and A.~J.~W. Hilton.
\newblock The chromatic index of graphs of even order with many edges.
\newblock {\em Journal of Graph Theory}, 8(4):463--470, 1984.

\bibitem{CH84}
A.~G. Chetwynd and A.~J.~W. Hilton.
\newblock The chromatic index of graphs with at most four vertices of maximum
  degree.
\newblock {\em Congressus Numerantium}, 43:221--248, 1984.

\bibitem{CH85}
A.~G. Chetwynd and A.~J.~W. Hilton.
\newblock Regular graphs of high degree are $1$-factorizable.
\newblock {\em Proceedings of the London Mathematical Society}, 50:193--206,
  1985.

\bibitem{CH89}
A.~G. Chetwynd and A.~J.~W. Hilton.
\newblock A delta-subgraph condition for a graph to be class 1.
\newblock {\em J. Comb. Theory, Ser. {B}}, 46(1):37--45, 1989.

\bibitem{CH90}
A.~G. Chetwynd and A.~J.~W. Hilton.
\newblock The chromatic index of graphs with large maximum degree, where the
  number of vertices of maximum degree is relatively small.
\newblock {\em Journal of Combinatorial Theory, Series {B}}, 48(1):45--66,
  1990.

\bibitem{Ch96}
K.~H. Chew.
\newblock The chromatic index of graphs of high maximum degree.
\newblock {\em Discrete Mathematics}, 162(1-3):77--91, 1996.

\bibitem{AMM15}
S.~M. de~Almeida, C.~P. de~Mello, and A.~Morgana.
\newblock Edge-coloring of split graphs.
\newblock {\em Ars Combinatoria}, 119:363--375, 2015.

\bibitem{SG07}
C.~de~Simone and A.~Galluccio.
\newblock Edge-colouring of regular graphs of large degree.
\newblock {\em Theor. Comput. Sci.}, 389(1-2):91--99, 2007.

\bibitem{EFK84}
A.~Ehrenfeucht, V.~Faber, and H.~Kierstead.
\newblock A new method of proving theorems on chromatic index.
\newblock {\em Discrete Mathematics}, 52(2):159 -- 164, 1984.

\bibitem{Fo77}
J.~Fournier.
\newblock M\'ethode et th\'eor\`eme g\'en\'erale de coloration des ar\^etes.
\newblock {\em Journal de Math\'ematiques Pures et Appliqu\'ees}, 56:437--453,
  1977.

\bibitem{GLPR}
E.~Galby, P.~T. Lima, D.~Paulusma, and B.~Ries.
\newblock Classifying $k$-edge colouring for ${H}$-free graphs.
\newblock {\em CoRR}, abs/1810.04379, 2018.

\bibitem{HZ93}
A.~J.~W. Hilton and C.~Zhao.
\newblock A sufficient condition for a regular graph to be class 1.
\newblock {\em Journal of Graph Theory}, 17(6):701--712, 1993.

\bibitem{Ho81}
I.~Holyer.
\newblock The {N}{P}-completeness of edge-coloring.
\newblock {\em {SIAM} Journal on Computing}, 10(4):718--720, 1981.

\bibitem{LG83}
D.~Leven and Z.~Galil.
\newblock {N}{P} completeness of finding the chromatic index of regular graphs.
\newblock {\em Journal of Algorithms}, 4(1):35--44, 1983.

\bibitem{LGZA15}
A.~R.~C. Lima, G.~Garcia, L.~Zatesko, and S.~M. de~Almeida.
\newblock On the chromatic index of cographs and join graphs.
\newblock {\em Eletronic Notes in Discrete Mathematics}, 50:433--438, 2015.

\bibitem{MF10}
R.~C.~S. Machado and C.~M.~H. de~Figueiredo.
\newblock Decompositions for edge-coloring join graphs and cobipartite graphs.
\newblock {\em Discrete Applied Mathematics}, 158:1336--1342, 2010.

\bibitem{MFT13}
R.~C.~S. Machado, C.~M.~H. de~Figueiredo, and N.~Trotignon.
\newblock Edge-colouring and total-colouring chordless graphs.
\newblock {\em Discrete Mathematics}, 313(14):1547--1552, 2013.

\bibitem{Ma14}
D.~S. Malyshev.
\newblock The complexity of the edge 3-colorability problem for graphs without
  two induced fragments each on at most six vertices.
\newblock {\em Sib. elektr. matem. izv.}, 11:811--822, 2014.

\bibitem{Ma17}
D.~S. Malyshev.
\newblock Complexity classification of the edge coloring problem for a family
  of graph classes.
\newblock {\em Discrete Mathematics and Applications}, 27:97--101, 2017.

\bibitem{MG92}
J. Misra and D. Gries.
A Constructive Proof of Vizing's Theorem. Information Processing Letters 41(3):131--133, 1992.

\bibitem{NV90}
T.~Niessen and L.~Volkmann.
\newblock Class 1 conditions depending on the minimum degree and the number of
  vertices of maximum degree.
\newblock {\em Journal of Graph Theory}, 14(2):225--246, 1990.

\bibitem{OMS98}
C.~Ortiz, N.~Maculan, and J.~L. Szwarcfiter.
\newblock Characterizing and edge-colouring split-indifference graphs.
\newblock {\em Discrete Applied Mathematics}, 82(1-3):209--217, 1998.

\bibitem{Vi64}
V.~G. Vizing.
\newblock On an estimate of the chromatic class of a $p$-graph.
\newblock {\em Diskret. Analiz.}, 3:25--30, 1964.

\bibitem{Ya86}
H.~P. Yap.
\newblock {\em Some Topics in Graph Theory}.
\newblock Cambridge University Press, New York, NY, USA, 1986.

\bibitem{Za18}
L.M. Zatesko. Novel Procedures for Graph Edge-Coloring, PhD Thesis, Federal University of Paran\'a,  2018.

\end{thebibliography}
\end{document}